\documentclass[12pt]{elsarticle}

\usepackage{amssymb}
\usepackage[utf8]{inputenc}
\usepackage{microtype}
\usepackage[colorlinks,allcolors=blue]{hyperref}
\usepackage{graphicx}
\usepackage{amsmath,amssymb,amsthm}
\usepackage{pifont}
\usepackage{enumitem}
\newcommand{\dpar}[1]{\left(#1\right)}
\newcommand{\dsqr}[1]{\left[#1\right]}

\newcommand{\E}{\mathbb{E}}
\usepackage[margin=2.5cm]{geometry}

\newtheorem*{thm}{Theorem}

\journal{Physica A}

\begin{document}

\begin{frontmatter}



\title{Characterization of Time Series Via R\'enyi Complexity-Entropy Curves}

\author[uem]{M. Jauregui}
\author[zunino1,zunino2]{L. Zunino}
\author[ekl]{E. K. Lenzi}
\author[uem]{R. S. Mendes}
\author[uem]{H. V. Ribeiro}\ead{hvr@dfi.uem.br}

\address[uem]{Departamento de Física, Universidade Estadual de Maringá, Maringá, PR 87020-900, Brazil}

\address[zunino1]{Centro de Investigaciones \'Opticas (CONICET La Plata - CIC), C.C. 3, 1897 Gonnet, Argentina}

\address[zunino2]{Departamento de Ciencias B\'asicas, Facultad de  Ingenier\'ia, Universidad Nacional de La Plata (UNLP), 1900 La  Plata, Argentina}

\address[ekl]{Departamento de F\'isica, Universidade Estadual de Ponta Grossa, Ponta Grossa, PR 84030-900, Brazil}

\begin{abstract}
One of the most useful tools for distinguishing between chaotic and stochastic time series is the so-called complexity-entropy causality plane. This diagram involves two complexity measures: the Shannon entropy and the statistical complexity. Recently, this idea has been generalized by considering the Tsallis monoparametric generalization of the Shannon entropy, yielding complexity-entropy curves. These curves have proven to enhance the discrimination among different time series related to stochastic and chaotic processes of numerical and experimental nature. Here we further explore these complexity-entropy curves in the context of the R\'enyi entropy, which is another monoparametric generalization of the Shannon entropy. By combining the R\'enyi entropy with the proper generalization of the statistical complexity, we associate a parametric curve (the  R\'enyi complexity-entropy curve) with a given time series. We explore this approach in a series of numerical and experimental applications, demonstrating the usefulness of this new technique for time series analysis. We show that the R\'enyi complexity-entropy curves enable the differentiation among time series of chaotic, stochastic, and periodic nature. In particular, time series of stochastic nature are associated with curves displaying positive curvature in a neighborhood of their initial points, whereas curves related to chaotic phenomena have a negative curvature; finally, periodic time series are represented by vertical straight lines.
\end{abstract}


\end{frontmatter}

\section{Introduction}
Quantifying the degree of complexity of a system is a common task when studying the most diverse complex systems. This task usually starts by constructing a time series and then considering a complexity measure. Since there is an inherent difficulty in defining the concept of complexity, researchers have employed several approaches/theories as possible complexity measures. A non-exhaustive list includes entropies~\cite{Shannon1948}, relative entropies~\cite{KullbackLeibler1951}, algorithmic complexities~\cite{Kolmogorov1965}, fractal dimensions~\cite{Mandelbrot}, Lyapunov exponents~\cite{Lyapunov}, and other tradicional nonlinear time series methods~\cite{perc2005nonlinear}. However, most of these approaches suffer from the drawbacks of being strongly sensitive on tuning parameters, hindering the reproducibility of results.

A possible way to overcome these difficulties is to employ the permutation entropy,
introduced by Bandt and Pompe~\cite{BandtPompe2002}. This complexity measure is basically the Shannon entropy of the distribution of the permutations associated with $d$-dimensional partitions $(x_k,x_{k+1},\ldots,x_{k+d-1})$ of a time series $(x_1,\ldots,x_m)$. It is common to choose $d\in\{3,4,5,6,7\}$ in most practical applications, in such a way that the number of permutations $d!$ is much smaller than $m$. For this reason, the computational cost of computing the permutation entropy is usually lower than the ones related to other complexity measures. Also, the idea of associating permutations with finite-dimensional partitions of a time series allows the application of this method to time series of arbitrary nature. These remarks agree with the fact that the method of Bandt and Pompe is already widely spread over the scientific community~\cite{Bian_etal2012,ribeiro2012complexity,Ribeiro_etal2012,Aragoneses_etal2013,LiZuntao2014,Weck_etal2015,Yang_etal2015,Aragoneses_etal2016,LinKhurramHong2016,Zunino2016679}.

In some cases, the permutation entropy can distinguish among time series of regular, chaotic and stochastic behavior. However, Rosso~\textit{et al.}~\cite{RossoLarrondoMartinPlastinoFuentes2007} have demonstrated that this complexity measure alone is not enough for properly performing this task. For instance, they have shown that time series related to the logistic map at fully developed chaos and time series associated with power-law correlated noises can display practically the same value of permutation entropy. Mainly because of that, Rosso~\textit{et al.} have employed the joint use of the permutation entropy and another complexity measure, called the statistical complexity~\cite{LopezManciniCalbet1995,anteneodo1996some,LambertiMartinPlastinoRosso2004}. The statistical complexity is basically the product of the permutation entropy by a distance between the distribution of the permutations and the uniform distribution. Having the values of the permutation entropy $H$ and the statistical complexity $C$ associated with a given time series, Rosso \textit{et al.} have represented this series by a point $(H,C)$ in a diagram of $C$ versus $H$. This diagram is the so-called complexity-entropy causality plane, where the term causality refer to the fact that temporal correlations are taken into account by the Bandt and Pompe approach. In this representation space, time series of chaotic and stochastic nature are represented by points located in different regions, that is, noise and chaos can be distinguished by using the complexity-entropy causality plane. 

However, we have recently depicted several situations in which the values of $H$ and $C$ are not enough for distinguishing among time series of distinct nature~\cite{RibeiroJaureguiZuninoLenzi2017}; for instance, the points $(H,C)$ can become very close to each other for time series displaying different periodic and chaotic behaviors. Motivated by this fact, we have extended the causality plane for considering the Tsallis~\cite{Tsallis1988} monoparametric generalization of the Shannon entropy~\cite{RibeiroJaureguiZuninoLenzi2017}. The values of the parameter of the Tsallis entropy give different weights to the probabilities associated with the permutations; consequently, different dynamical scales of the system are accessed by varying the entropy parameter. In that article, we associated parametric curves with time series based on the different values of $(H,C)$ obtained by changing the Tsallis entropy parameter, a representation that we call the complexity-entropy curve. These curves have proven to enhance the differentiation of time series of regular, chaotic and stochastic nature even in cases in which the usual complexity-entropy causality plane does not provide useful information. 

On the other hand and similarly to what happens with the concept of complexity, there are several other entropy definitions in the context of information theory~\cite{gray2011entropy,cover2012elements,beck2009generalised}. These different entropies allow us to explore, capture and quantify different forms of complexity, leading us to more suitable descriptions for characterizing the most diverse complex systems addressed by physicists. Here we further explore this idea by considering the R\'enyi entropy~\cite{renyi1961} in place of the Tsallis entropy~\cite{Tsallis1988}. The R\'enyi entropy is also a monoparametric generalization of the Shannon entropy, which has been employed in several contexts such as medical/diagnostics applications~\cite{kannathal2005entropies}, time-frequency analysis~\cite{baraniuk2001measuring}, quantum entanglement measures~\cite{portesi1996generalized,islam2015measuring}, and image thresholding~\cite{sahoo2004thresholding}. Therefore, in analogy to the Tsallis entropy case, we shall associate a parametric curve with a given time series (the \textit{R\'enyi complexity-entropy curve}), and by exploring some properties of this curve, we can characterize the time series under study. Among other findings, we show that the curvature of these curves identifies whether a time series is of a stochastic or a chaotic nature, and that periodic time series are represented by vertical straight lines.

The organization of the article is as follows. Section~\ref{sec:defin-stat-compl} provides the definitions of the R\'enyi entropy and the R\'enyi statistical complexity. Section~\ref{sec:renyi-compl-entr} gives a brief description of the method of Bandt and Pompe for defining the ordinal probabilities from a given time series. We also work out a list of general properties of the R\'enyi complexity-entropy curves. In Section~\ref{sec:char-some-time}, we analyze several time series of chaotic and stochastic nature, obtained by numerical procedures or by experimental measurements. Finally, we conclude in Section~\ref{sec:conclusions}.

\section{A definition of a statistical complexity based on the R\'enyi entropy}
\label{sec:defin-stat-compl}
The R\'enyi entropy of a discrete probability distribution $p=(p_1,\ldots,p_n)$ is defined as~\cite{renyi1961}
\begin{equation}
  S_\alpha(p)=\frac{1}{1-\alpha}\ln\sum_{i=1}^np_i^\alpha\,,\quad \alpha>0\,,\,\alpha\ne 1\,.
\end{equation}
From this definition, we immediately note that $S_\alpha(p)$ recovers the Shannon entropy of $p$ when $\alpha$ tends to $1$. We can further verify that the maximum value of the R\'enyi entropy is equal to $\ln n$ (as in the Shannon entropy case), which happens when the uniform distribution $u=(1/n,\ldots,1/n)$ is considered. This enables us to define the normalized R\'enyi entropy of $p$ as
\begin{equation}
  \label{eq:Halpha}
  H_\alpha(p)=\frac{S_\alpha(p)}{\ln n}\,.
\end{equation}

By following Martin, Plastino and Rosso~\cite{MartinPlastinoRosso2006}, we define the R\'enyi statistical complexity of $p$ as
\begin{equation}
  \label{eq:Calpha}
  C_\alpha(p)=\frac{D_\alpha(p)H_\alpha(p)}{D_\alpha^*}\,,
\end{equation}
where
\begin{multline}
  \label{eq:Dalpha}
  D_\alpha(p)=\frac{1}{2(\alpha-1)}\left[\ln\sum_{i=1}^np_i^\alpha\dpar{\frac{p_i+1/n}{2}}^{1-\alpha}\right.\\
  \left.+\ln\sum_{i=1}^n\frac{1}{n^\alpha}\dpar{\frac{p_i+1/n}{2}}^{1-\alpha}\right]
\end{multline}
and
\begin{equation}
  \label{eq:Dmax}
  D_\alpha^*=\frac{1}{2(\alpha-1)}\ln\dsqr{\frac{(n+1)^{1-\alpha}+n-1}{n}\dpar{\frac{n+1}{4n}}^{1-\alpha}}\,.
\end{equation}
The quantity $D_\alpha(p)$ is always non-negative~\cite{ErvenHarremos2014} and can be interpreted as a distance between the distribution $p=(p_1,\ldots,p_n)$ and the uniform distribution $u=(1/n,\ldots,1/n)$. In fact, $D_\alpha(p)$ is a generalization of the Jensen-Shannon divergence of $p$ and $u$~\cite{Lin1991}. The quantity $D_\alpha^*$ represents the maximum value of $D_\alpha(p)$, which is reached when $p$ has only one non-zero component. By combining this fact with Eq.~(\ref{eq:Calpha}), we have that $0\le C_\alpha(p)\le 1$.

\section{R\'enyi complexity-entropy curves in the framework of Bandt and Pompe}
\label{sec:renyi-compl-entr} 
We start by briefly describing the method of Bandt and Pompe~\cite{BandtPompe2002} for defining the ordinal probabilities from a given time series $(x_1,\ldots,x_m)$. Fixed an integer $d>1$ such that $d!\ll m$, we consider the set $E$ of all $d$-dimensional vectors $(x_k,\ldots,x_{k+d-1})$, where $k=1,\ldots,m-d+1$. We next define a mapping $\Pi$ from $E$ into the set of all permutations of the set $\{0,\ldots,d-1\}$ such that $\Pi(x_k,\ldots,x_{k+d-1})=\pi$, where the permutation $\pi$, which can be represented as a vector $(\pi(0),\ldots,\pi(d-1))$, satisfies the conditions
\begin{enumerate}
\item[(i)] $x_{k+\pi(0)}\le\ldots\le x_{k+\pi(d-1)}$;
\item[(ii)] if $x_{k+\pi(j)}=x_{k+\pi(j+1)}$, then $\pi(j)<\pi(j+1)$.
\end{enumerate}
Finally, we define the probabilities
\begin{equation}
  \label{eq:BP}
  p(\pi)=\frac{\#\{v\in E:\Pi(v)=\pi\}}{m-d+1}\,,
\end{equation}
where the symbol $\#$ denotes cardinality.

To illustrate the procedure described above, let us consider the hypothetical time series $(3,5,1,6,6,4)$. By choosing $d=2$, the vectors in $E$ are $(3,5)$, $(5,1)$, $(1,6)$, $(6,6)$ and $(6,4)$. Then, $\Pi(3,5)=\Pi(1,6)=\Pi(6,6)=(0,1)$ and $\Pi(5,1)=\Pi(6,4)=(1,0)$. Hence, we have the probabilities $p(0,1)=3/5$ and $p(1,0)=2/5$.

Having the probability distribution associated with the time series $(x_1,\ldots,x_m)$, we calculate the normalized R\'enyi entropy and the statistical complexity by using Eqs.~(\ref{eq:Halpha}) and~(\ref{eq:Calpha}) with $n=d!$. In this manner, for each embedding dimension $d$, we construct a parametric curve $C_\alpha(p)$ versus $H_\alpha(p)$, considering $\alpha$ as a real parameter that takes values in the interval $(0,\infty)$. We call these curves as the \textit{R\'enyi complexity-entropy curves} by analogy with the $q$-complexity-entropy curves obtained using the Tsallis entropy~\cite{RibeiroJaureguiZuninoLenzi2017}. Since
$H_\alpha(p)$ is a monotonically non-increasing function of $\alpha$~\cite{ErvenHarremos2014,BeckSchogl}, R\'enyi complexity-entropy curves are never closed, \textit{i.e.}, they do not form loops, in contrast with $q$-complexity-entropy curves, which are likely to form loops for time series related to stochastic processes.

We can prove the following general properties of the R\'enyi complexity-entropy curves associated with an arbitrary time series in the framework of Bandt and Pompe, considering an embedding dimension $d$ (see~\ref{sec:limit-values} for details):
\begin{enumerate}
\item[(i)] If only one permutation occurs --- for instance, if the time series is strictly monotonic --- the R\'enyi complexity-entropy curve reduces to the single point~$(0,0)$.
\item[(ii)] If all the allowed permutations occur, the R\'enyi complexity-entropy curve begins at the point $(1,0)$, corresponding to $\alpha\downarrow 0$ (that is, when $\alpha$ tends to zero from the right), and ends at the point $(h_f,c_f)$, corresponding to $\alpha\to\infty$. From $h_f$, we obtain the value of the maximum component $p_M$ of the probability distribution by using the simple relation $p_M=(d!)^{-h_f}$. The ending point $(h_f,c_f)$ also gives us the value of the minimum component $p_m$ of the probability distribution by the relation
\begin{equation}
  p_m=\frac{4p_M}{d!p_M+1}\dpar{\frac{d!+1}{4d!}}^{c_f/h_f}-\frac{1}{d!}\,.
\end{equation}
\item[(iii)] If $r$ permutations occur, with $1<r<d!$, then the R\'enyi complexity-entropy curve begins at a point $(h_i,c_i)$ and ends at a point $(h_f,c_f)$. The starting point gives us the value of $r$ by means of the simple relation $r=(d!)^{h_i}$. From the ending point, we obtain the value of the maximum component $p_M$ of the probability distribution by using the relation $p_M=(d!)^{-h_f}$. In this case, the minimum component of $p$ is clearly zero.
\end{enumerate}

An interesting consequence of item (iii) happens when the $r$ permutations that actually occur have the same probability, namely $1/r$. In this case, we have that $p_M=1/r$, and from $r=(d!)^{h_i}$ and $p_M=(d!)^{-h_f}$, we obtain that $h_f=h_i=\ln r/\ln d!$. Thus, the corresponding R\'enyi complexity-entropy curve is a vertical straight line --- we can verify straightforwardly that the statistical complexity still depends on the parameter $\alpha$ in this case. This situation can arise from the analysis of a time series displaying a periodic behavior.

\section{Characterization of time series via R\'enyi complexity-entropy curves}
\label{sec:char-some-time}
In this section, we explore the R\'enyi complexity-entropy curves associated with several time series using the procedure described in the previous section. In our analysis, we consider time series obtained by numerical procedures (such as chaotic maps) and experimental measurements (such as fluctuations of crude oil prices).

\subsection{Fractional Brownian motions}
A stochastic process $(B^{\cal H}_t)_{t\ge 0}$, characterized by a parameter ${\cal H}\in (0,1)$, is called a fractional Brownian motion~\cite{Mandelbrot1968,Shevchenko2015} if it is a Gaussian process with null expectation and covariance function
\begin{equation}
  \E(B^{\cal H}_tB^{\cal H}_s)=\frac{1}{2}(t^{2{\cal H}}+s^{2{\cal H}}-|t-s|^{2{\cal H}})\,.
\end{equation}
The parameter ${\cal H}$ is usually called the Hurst parameter or the Hurst exponent. If ${\cal H}=1/2$, the increments $B^{\cal H}_{t_2}-B^{\cal H}_{s_2}$ and
$B^{\cal H}_{t_1}-B^{\cal H}_{s_1}$, with $s_1<t_1<s_2<t_2$, are independent and the usual Brownian motion is recovered. If ${\cal H}>1/2$ (${\cal H}<1/2$), these increments are positively (negatively) correlated. This means, roughly speaking, that positive increments in the past are likely to generate positive (negative) increments in the future and vice versa. Moreover, if $1/2<{\cal H}<1$, the fractional Brownian motion exhibits long-range correlations, in the sense that~\cite{Wijeratne2015}
\begin{equation}
  \sum_{n=1}^\infty|\E((B^{\cal H}_{n+1}-B^{\cal H}_n)(B^{\cal H}_1-B^{\cal H}_0))|=\infty\,.
\end{equation}

We numerically generate time series of length $2^{17}$ from fractional Brownian motions with different Hurst exponents using Hosking's procedure~\cite{Hosking}. Figure~\ref{fig:fbm} shows the R\'enyi complexity-entropy curves related to fractional Brownian motions with Hurst exponents ${\cal H}\in\{0.1,\ldots,0.8\}$. In the left panels, an embedding dimension $d=3$ has been considered whereas $d=4$ has been used for the right panels. Each colored curve in Fig.~\ref{fig:fbm} represents the mean R\'enyi complexity-entropy curve over $100$ realizations associated with a given Hurst parameter, and the shaded areas are obtained by considering two standard deviations in both the $H_\alpha$ and $C_\alpha$ values. The dashed lines were obtained using the analytical expressions for the probabilities of the permutations when $d=3$ and $d=4$, obtained by Bandt and Shiha~\cite{BandtShiha2007} (see also Appendix B of Ref.~\cite{RibeiroJaureguiZuninoLenzi2017}).

\begin{figure}[!ht]
  \centering
  \includegraphics[width=0.8\textwidth,keepaspectratio]{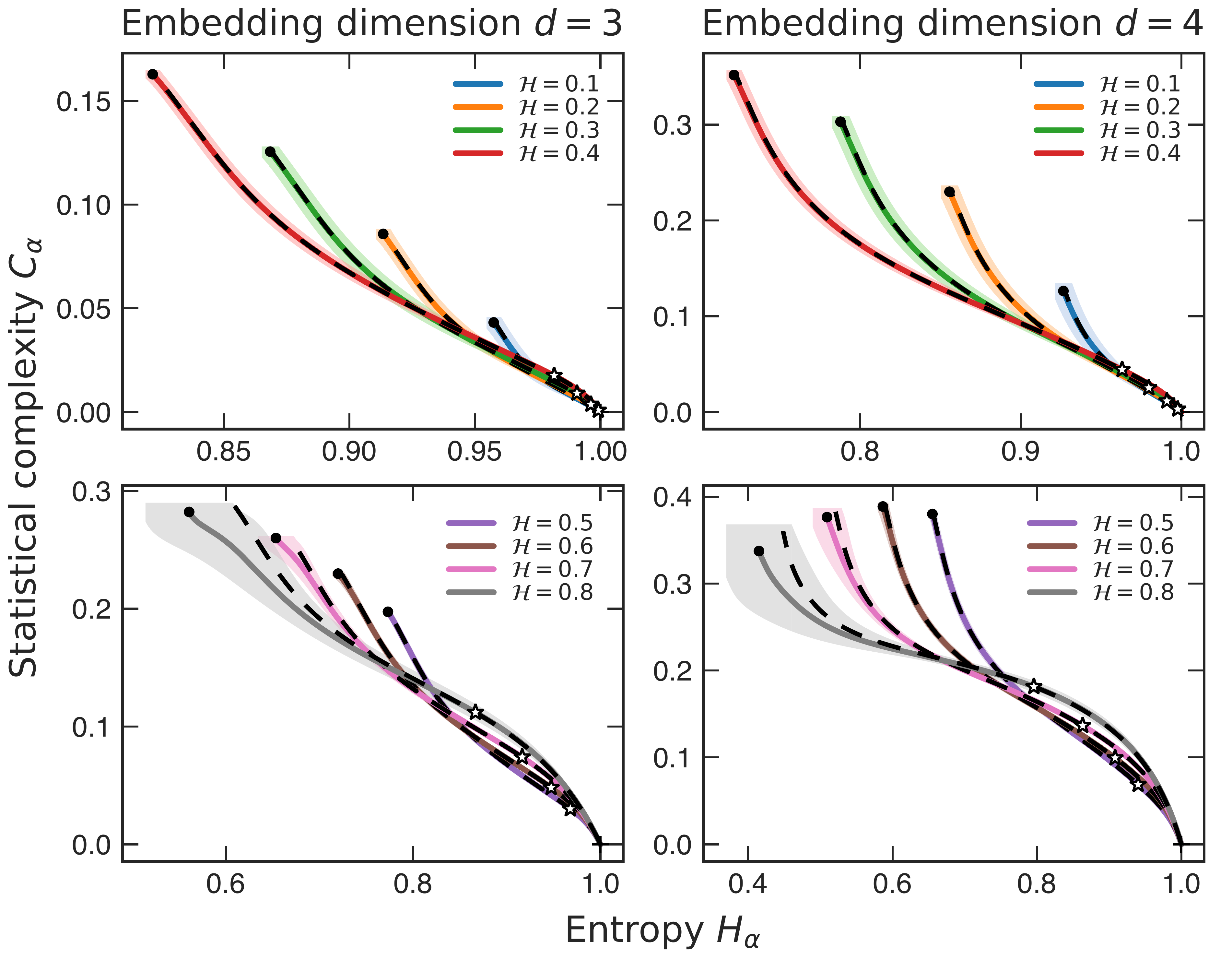}
  \caption{The R\'enyi complexity-entropy curves related to fractional Brownian motions with Hurst exponents $\mathcal{H}\in\{0.1,\ldots,0.8\}$. We have considered the embedding dimensions $d=3$ for the left panels, and $d=4$ for the right panels. The colored lines represent the mean R\'enyi complexity-entropy curves over $100$ realizations, and the shaded areas are obtained by considering two standard deviations in both the $H_\alpha$ and $C_\alpha$ values. The dashed lines are the R\'enyi complexity-entropy curves obtained analytically using the exact probabilities given in Ref.~\cite{BandtShiha2007}. The markers~$+$ and~$\circ$ indicate the beginning and the end of each curve respectively. The~\ding{73} markers indicate the ordered pair associated with the usual permutation entropy and statistical complexity values obtained when $\alpha$ tends to~$1$.}
  \label{fig:fbm}
\end{figure}

We note immediately from Fig.~\ref{fig:fbm} that all R\'enyi complexity-entropy curves start at the point $(1,0)$, indicating that all permutations occur for embedding dimensions $d=3$ and $d=4$. We further observe a good agreement between the numerically obtained curves and the exact results. Another common characteristic of these curves is that they have a positive curvature in a neighborhood of the starting point $(1,0)$. In other words, the derivative $dC_\alpha/dH_\alpha$ is a monotonic increasing function of $\alpha$ in a neighborhood of $\alpha=0$, as shown in Fig.~\ref{fig:fbmD}. Putting $H_\infty=\lim_{\alpha\to\infty}H_\alpha$, we observe from Fig.~\ref{fig:fbm} that $H_\infty$ is a monotonically decreasing function of the Hurst parameter. Moreover, the points associated with these values (indicated by $\circ$ markers) are more distant from each other than the ones related to the usual Shannon entropy (indicated by \ding{73} markers). Thus, these points associated with $H_\infty$ enable a better differentiation of time series related to fractional Brownian motions with different Hurst exponents, as it was also discussed in Ref.~\cite{zunino2015permutation}.

\begin{figure}[!ht]
  \centering \includegraphics[width=0.8\textwidth,
  keepaspectratio]{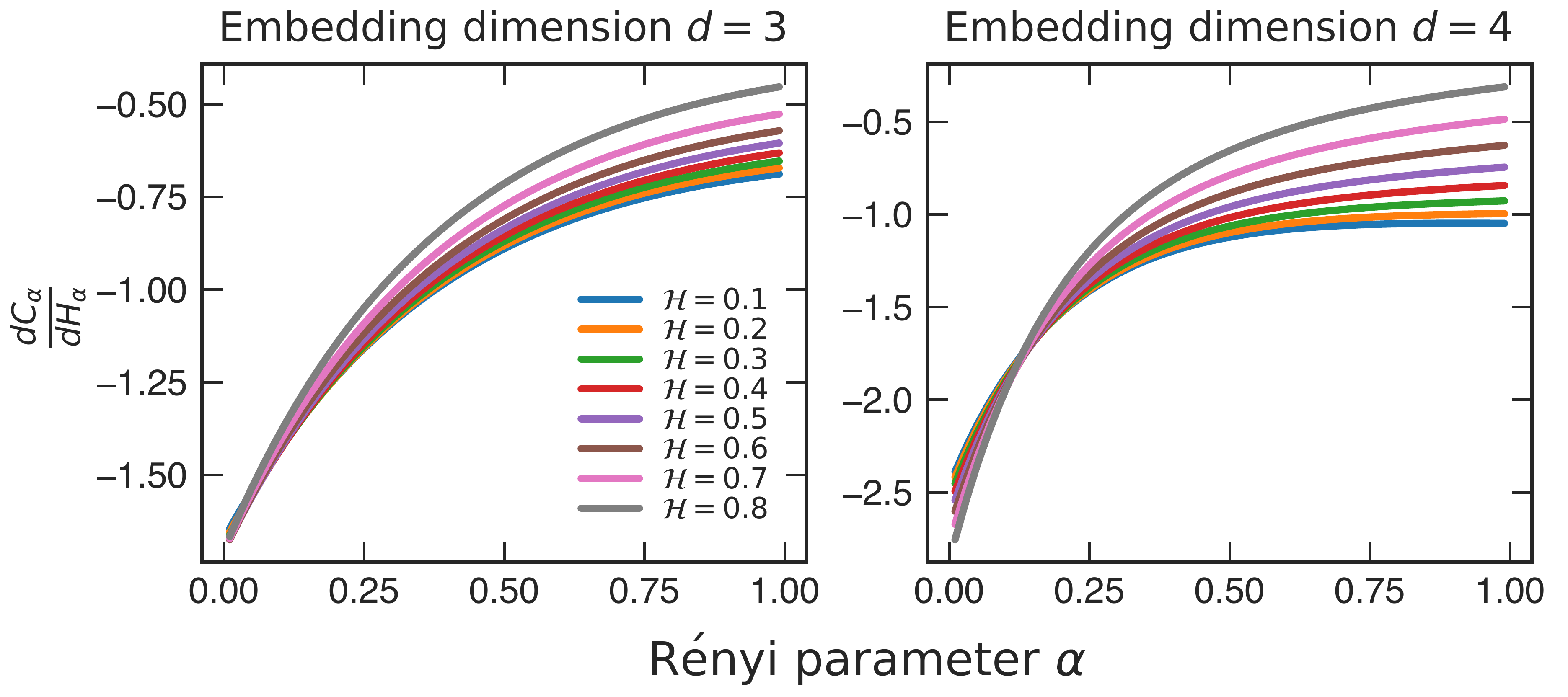}
  \caption{Representation of the derivative of the statistical complexity $C_\alpha$ with respect to the entropy $H_\alpha$ as a function of the R\'enyi parameter $\alpha$ for the fractional Brownian motions with Hurst parameter ${\cal H}\in\{0.1,\ldots,0.8\}$, the same considered in Fig.~\ref{fig:fbm}.}
  \label{fig:fbmD}
\end{figure}

\subsection{Chaotic maps at fully developed chaos}
In addition to time series obtained from stochastic processes, we also investigate time series associated with chaotic phenomena. In particular, we construct time series of length $10^4+2^{17}$ by the iteration of eight different chaotic maps, namely: Burgers, cubic, Gingerbreadman, Henon, logistic, Ricker, sine, and Tinkerbell maps, at fully developed chaos. Details about each map are provided in Appendix C of Ref.~\cite{RibeiroJaureguiZuninoLenzi2017}. For the two-dimensional chaotic maps of Burgers, Gingerbreadman, Henon and Tinkerbell maps, we have considered the time series of the square of the sum of the two components. To avoid any possible transient behavior, we have removed the first $10^4$ iterations in all simulations.

\begin{figure}[!ht]
  \centering \includegraphics[width=0.8\textwidth,
  keepaspectratio]{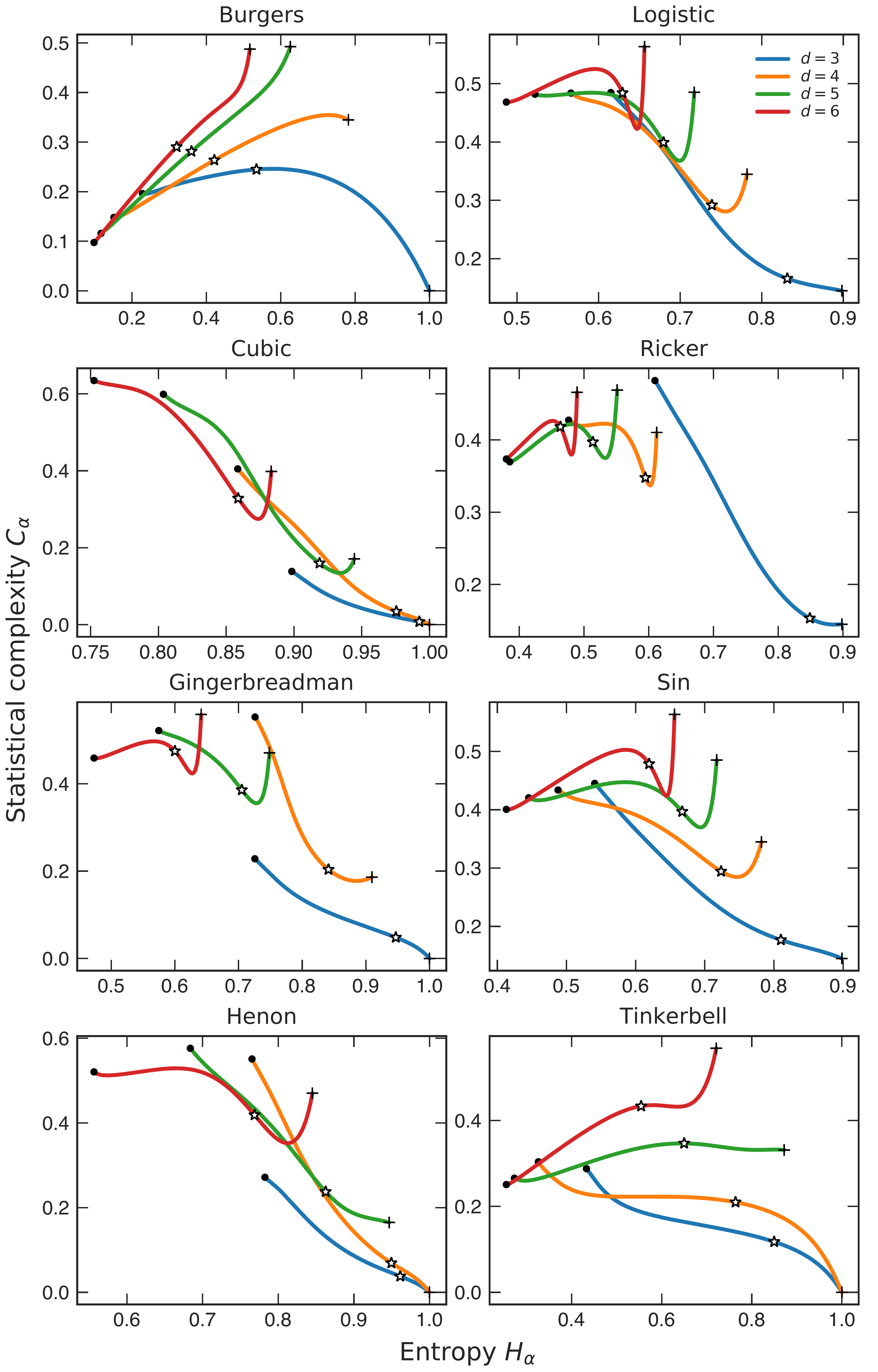}
  \caption{R\'enyi complexity-entropy curves of various chaotic maps at fully developed chaos. The embedding dimensions $d \in \{3,4,5,6\}$ were considered for each chaotic map. The markers~$+$ and~$\circ$ indicate the beginning and the end of each curve respectively. The~\ding{73} markers indicate the ordered pair associated with the usual permutation entropy and statistical complexity obtained when $\alpha$ tends to~$1$.}
  \label{fig:maps}
\end{figure}

Figure~\ref{fig:maps} shows the R\'enyi complexity-entropy curves associated with the eight chaotic maps mentioned in the previous paragraph. For each map, we have considered the embedding dimensions $d \in \{3,4,5,6\}$. In contrast with the curves associated with fractional Brownian motions, the R\'enyi complexity-entropy curves related to chaotic maps begin at points different from $(1,0)$, at least for embedding dimensions $d>4$. This feature indicates that there are permutations that do not occur. We further note that, at least for $d=6$, the R\'enyi complexity-entropy curves have negative curvature in a neighborhood of their initial point, \textit{i.e.}, the derivative $dC_\alpha/dH_\alpha$ is a decreasing function of the R\'enyi parameter $\alpha$ in the neighborhood of $\alpha=0$, as shown in Fig.~\ref{fig:mapsD}.

\begin{figure}[!ht]
  \centering \includegraphics[width=0.8\textwidth,
  keepaspectratio]{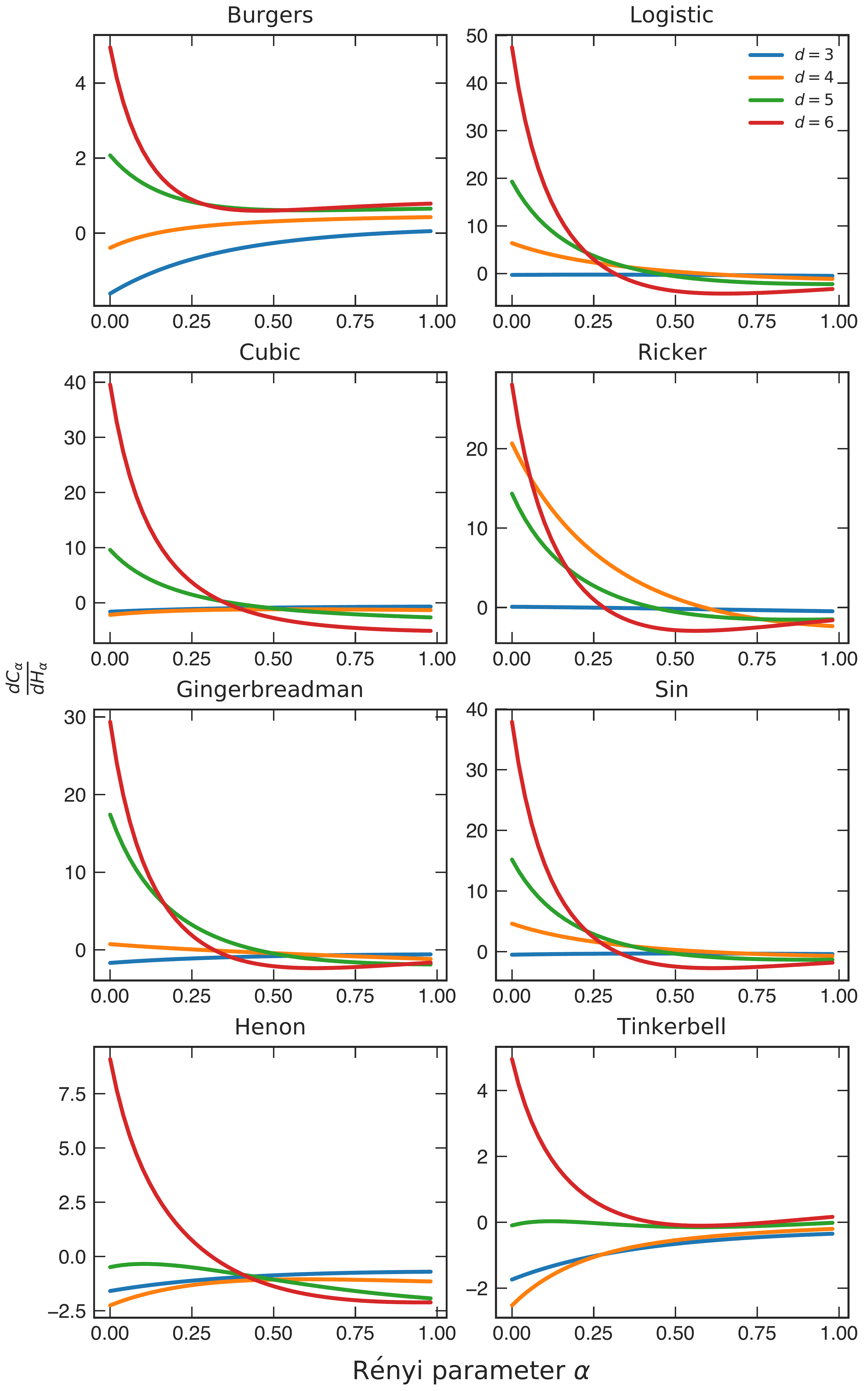}
  \caption{Representation of the derivative of the statistical complexity $C_\alpha$ with respect to the entropy $H_\alpha$ as a function of the R\'enyi parameter $\alpha$ for the chaotic maps considered in Fig.~\ref{fig:maps}.}
  \label{fig:mapsD}
\end{figure}
\clearpage
\subsection{The logistic map}
Among the chaotic maps considered in the previous subsection, the logistic map is undoubtedly one of the most famous. This map is defined by the recurrence formula
\begin{equation}
  y_{k+1}=ay_k(1-y_k)\,,
\end{equation}
where $a$ is a real parameter, whose values of interest are in the interval $[0,4]$. A modification of the value of $a$ changes the behavior of the logistic map. For instance, this map exhibits simple periodic behavior for $a = 3.05$, stable cycles of period $4$ for $a=3.5$ and of period $8$ for $a=3.55$, and chaos for most values of $a>3.569 945 67\ldots$ and for $a=4$ (fully developed chaos).

Figure~\ref{fig:logistic}a shows the R\'enyi complexity-entropy curves associated with the logistic map for $a\in\{3.05,3.5,3.55,3.593,4\}$ and embedding dimension $d=4$. We note that all curves begin at points different from $(1,0)$, indicating the lack of some permutations. Moreover, the R\'enyi complexity-entropy curves for $a=3.593$ and $a=4$, which correspond to two chaotic regimes of the logistic map, show negative curvature in a neighborhood of their starting points. On the other hand, for the other values of $a$, for which the logistic map has a periodic behavior, the R\'enyi complexity-entropy curves are almost vertical lines. This indicates that, for each $a\in\{3.05,3.5,3.55\}$, the permutations that actually occur have the same probability. This can be verified directly using the representation of the corresponding time series. Moreover, the R\'enyi complexity-entropy curves for $a=3.5$ and $a=3.55$ coincide, suggesting that both time series have the same probability distribution of the permutations. In fact, a detailed inspection of these time series reveals that the permutations that actually occur are the same for both of them. Another curious fact from Fig.~\ref{fig:logistic}a is that the point associated with the usual permutation entropy and statistical complexity for $a=3.55$ and $a=3.593$ almost coincide even when these two values correspond to completely different regimes of the logistic map.

\begin{figure}[!ht]
  \centering
  \includegraphics[width=0.8\textwidth,keepaspectratio]{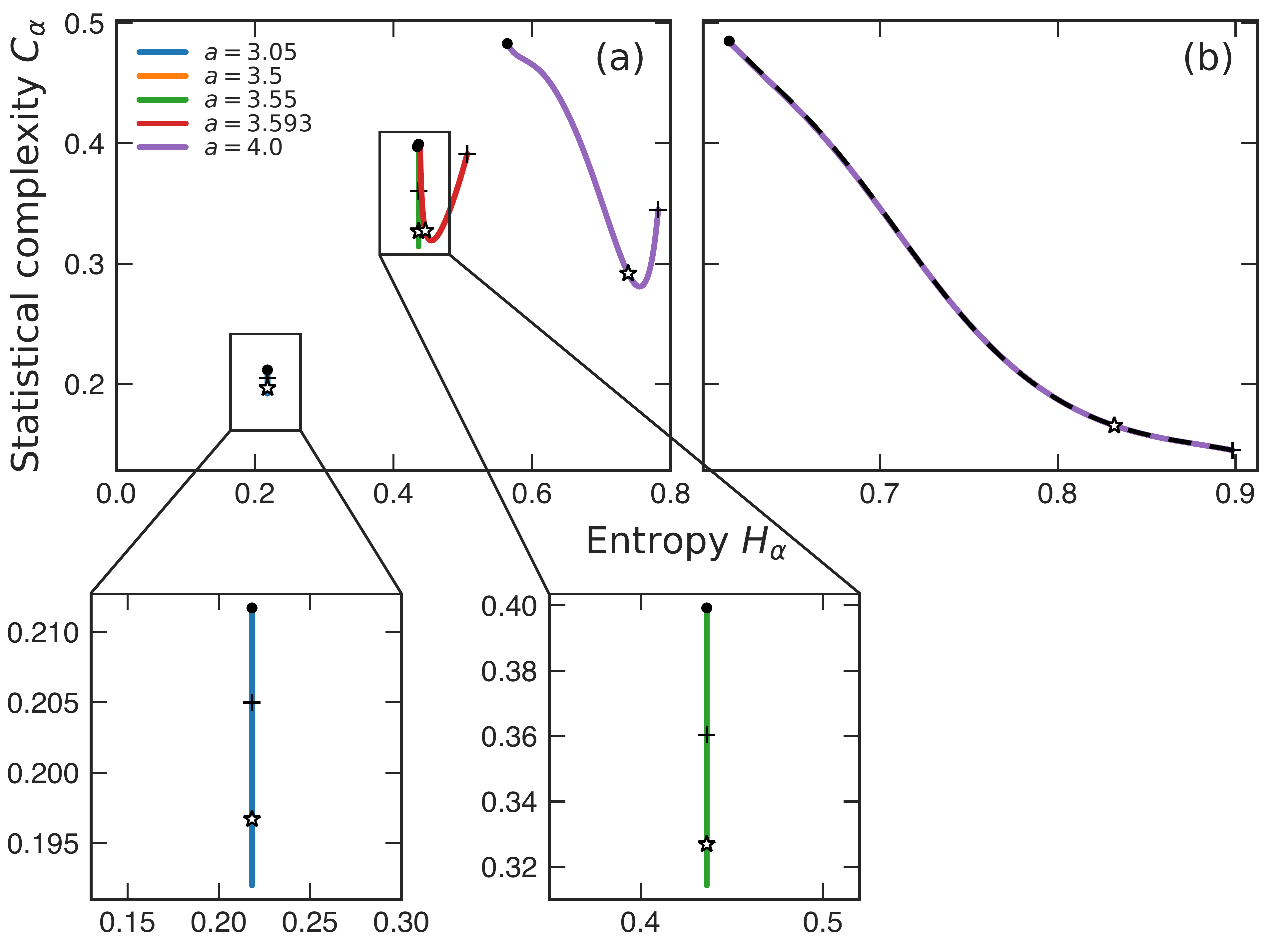}
  \caption{(a) The R\'enyi complexity-entropy curves associated with the logistic map for typical values of the parameter $a$ and embedding dimension $d=4$. The $+$ and $\circ$ markers represent the beginning and end of the curves respectively. The \ding{73} marker identifies the point associated with the usual permutation entropy and statistical complexity, obtained when $\alpha$ tends to $1$. (b) The R\'enyi complexity-entropy curve associated with the logistic map for $a=4$ and $d=3$. The dotted curve was analytically obtained using the exact values of the probabilities of the permutations. The meanings of the $+$, $\circ$ and \ding{73} markers are the same as in panel (a).}
  \label{fig:logistic}
\end{figure}

The logistic map specially attracts our interest because the probabilities of the permutations obtained from its time series within the Bandt and Pompe approach can be found analytically when $a=4$ and $d=3$. In fact, Amig\'o~\textit{et al.}~\cite{Amigo2006,Amigo2007,Amigo2008,Amigo} have shown that the list $(y_{k},y_{k+1},y_{k+2})$ corresponds to the permutation $(0,1,2)$ if $0<y_k<1/4$. In the same way, the permutation $(0,2,1)$ occurs if $1/4<y_k<\frac{5-\sqrt{5}}{8}$, $(2,0,1)$ if $\frac{5-\sqrt{5}}{8}<y_k<3/4$, $(1,0,2)$ if $3/4<y_k<\frac{5+\sqrt{5}}{8}$, $(1,2,0)$ if $\frac{5-\sqrt{5}}{8}<y_k<1$, and no list $(y_{k},y_{k+1},y_{k+2})$ corresponds to the permutation $(1,0,2)$. Combining this information with the fact that the logistic map with $a=4$ has an invariant distribution concentrated in the interval $(0,1)$ with density $\pi^{-1}y^{-1/2}(1-y)^{-1/2}$~\cite{jakobson1981absolutely}, we obtain the probabilities of the permutations by integrating this density in each interval for which each permutation occurs. Thus, we obtain the distribution $p=(1/3,1/15,4/15,2/15,1/5,0)$ (following the order of appearance of the permutations), from which we analytically compute $H_\alpha(p)$ and $C_\alpha(p)$. Figure~\ref{fig:logistic}b shows a comparison between the numerical and the analytical R\'enyi complexity-entropy curves associated with the logistic map for $a=4$ and embedding dimension $d=3$, where we note that both curves practically coincide.

\subsection{Empirical data}
The analyses performed in the previous subsections suggest that the R\'enyi complexity-entropy curves enable the distinction of time series associated with stochastic processes, chaotic phenomena, and periodic behaviors. More precisely, a positive curvature of this curve in a neighborhood of its starting point indicates that the associated time series is of a stochastic nature, whereas a negative curvature, at least for large embedding dimensions, expresses that the time series may be of a chaotic nature. Also, a vertical R\'enyi complexity-entropy curve indicates that the time series has a periodic behavior. To further explore these affirmations, we now consider time series obtained by experimental measurements. In particular, we consider the time series generated from the intensity pulsations of a laser~\cite{huebner1989dimensions}, which is of a chaotic nature, and the one associated with the fluctuations of the crude oil price, which is of a stochastic nature.

The time series associated with the chaotic intensity pulsations of a laser has $9093$ terms and is freely available on the Internet~\cite{tspred}. The time series related to the crude oil prices refers to the daily closing spot price of the West Texas Intermediate from January 2, 1986, to July 10, 2012. This time series has $7788$ terms and can also be obtained freely on the Internet~\cite{eia}. The R\'enyi complexity-entropy curves for both time series and their derivatives as functions of the R\'enyi parameter $\alpha$ are represented in Fig.~\ref{fig:expdata} for embedding dimensions $d\in\{3,4,5,6\}$. We note that the R\'enyi complexity-entropy curves associated with the intensity pulsations of a laser start at points different from $(1,0)$ for embedding dimensions $d>3$, indicating the lack of some permutations. Moreover, the curvature of these curves in a neighborhood of their initial point is negative for embedding dimensions $d\ge 5$, in agreement with the chaotic nature of the corresponding time series. The R\'enyi complexity-entropies associated with the fluctuations of the crude oil price begin at the point $(1,0)$, except for the embedding dimension $d=6$. A possible justification for this fact is that the length of the series is not great enough --- note that $6!=720$, which is not too much less than $7788$. On the other hand, the curvature of these curves in a neighborhood of their starting points is positive, in agreement with the stochastic nature of the corresponding time series.


\begin{figure}[!ht]
  \centering \includegraphics[width=0.9\textwidth,keepaspectratio]{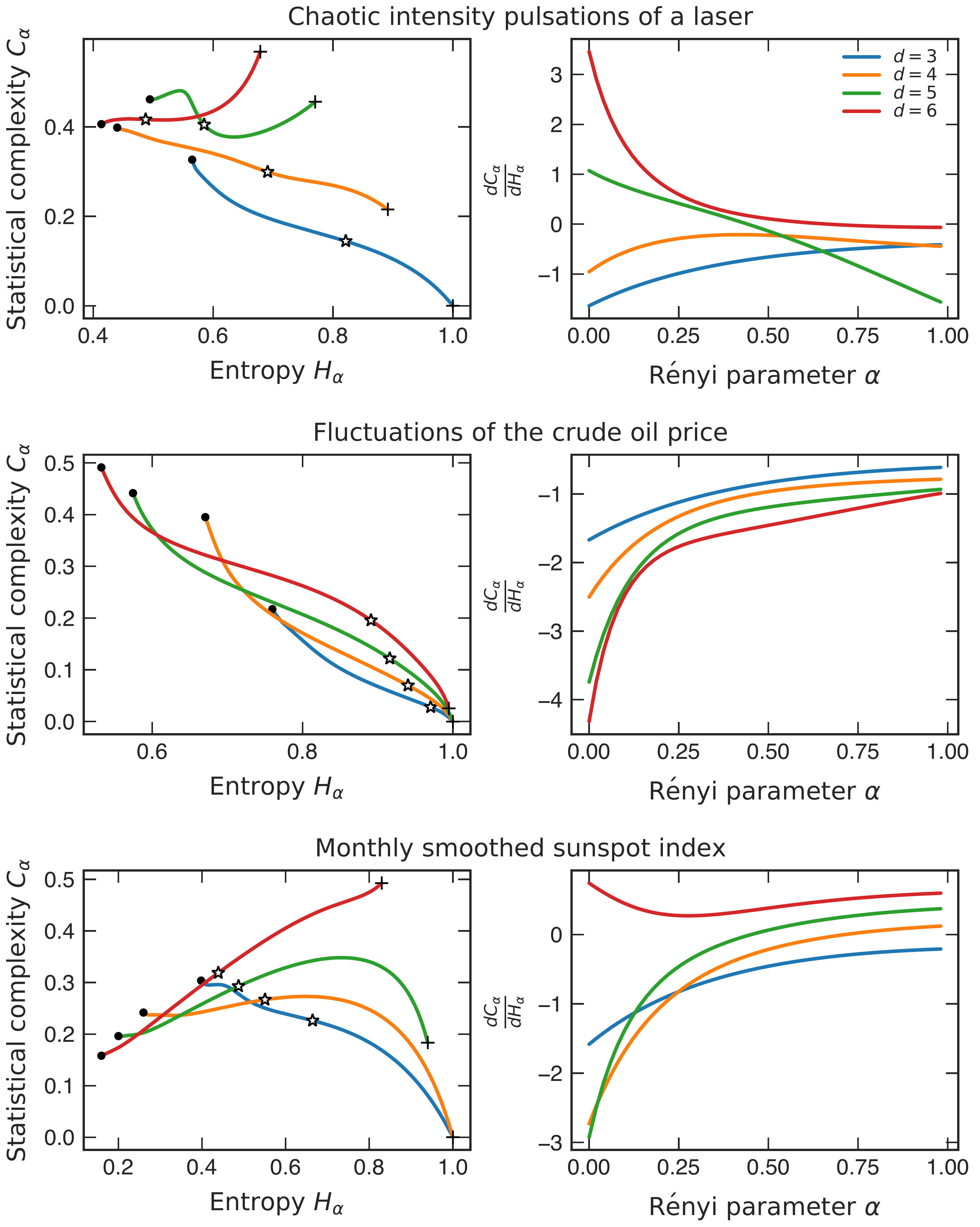}
  \caption{(Left panels) R\'enyi complexity-entropy curves for time series associated with the chaotic intensity pulsations of a laser, the fluctuations of the crude oil price, and the monthly smoothed sunspot index. The embedding dimensions $d \in \{3,4,5,6\}$ were considered for each time series. The $+$ and $\circ$ markers denote the beginning and the end of each curve. The \ding{73} markers represent the ordered pairs corresponding to the usual permutation entropy and statistical complexity, obtained when $\alpha$ tends to $1$. (Right panels) Representation of the derivative of $C_\alpha$ with respect to $H_\alpha$ as a function of the R\'enyi parameter $\alpha$ for the time series and embedding dimensions considered in the left panels.}
  \label{fig:expdata}
\end{figure}

In addition to time series with a well-defined nature, we also analyze the time series of the monthly smoothed sunspots index, for which there is still no consensus about its stochastic or chaotic nature~\cite{carbonell1993asymmetry,paluvs1999sunspot,timmer2000can,paluvs2000paluvs,mininni2000stochastic,mininni2002study,de2010fast}. This time series is freely available on the Internet~\cite{silso} and was generated by analyzing the $13$-month smoothed monthly sunspot index from 1974 to 2016, yielding a time series with $3202$ terms. The R\'enyi complexity-entropy-curves associated with this time series and their derivatives as functions of the R\'enyi parameter $\alpha$ for embedding dimensions $d\in\{3,4,5,6\}$ are represented in Fig.~\ref{fig:expdata}. We note that the R\'enyi complexity-entropy curves start at the point $(1,0)$, except for the embedding dimensions $d=5$ and $d=6$. Moreover, the curvature of these curves are positive for embedding dimensions $d\in\{3,4,5\}$ and is negative for $d=6$. For these reasons, we are not in a good position to state anything about the nature of the corresponding time series mainly because of its small length.

\section{Conclusions}
\label{sec:conclusions}
We have considered a generalized definition of the statistical complexity based on a monoparametric generalization of the Shannon entropy, namely the R\'enyi entropy. We have used this generalized statistical complexity $C_\alpha$ in combination with the normalized R\'enyi entropy $H_\alpha$ to construct a parametric curve $C_\alpha$ versus $H_\alpha$, taking the R\'enyi parameter $\alpha>0$ as the parameter of the curve. From this approach, we have analyzed several time series obtained from numerical simulations and experimental measurements. Our study has revealed that the use of R\'enyi complexity-entropy curves enables the differentiation of time series of chaotic and stochastic nature. More precisely, time series of stochastic nature are associated with R\'enyi complexity-entropy curves that have positive curvature in a neighborhood of their initial points. On the other hand, the curves related to chaotic phenomena have negative curvature in a neighborhood of their initial points, at least for large embedding dimensions.

The use of R\'enyi complexity-entropy curves enhances the detection of time series that have periodic behavior since for them the R\'enyi complexity-entropy curves are vertical lines. This feature is absent in the framework of $q$-complexity-entropy curves \cite{RibeiroJaureguiZuninoLenzi2017}, which are based on the Tsallis entropy. Thus, R\'enyi complexity-entropy curves may give complementary information about time series that are not properly taken into account by the $q$-complexity-entropy curves. 

\section*{Acknowledgments}
M.J. thanks the financial support of CNPq under Grant 150577/2017-6. H.V.R. acknowledges the financial support of CNPq under Grant 440650/2014-3. L.Z. acknowledges Consejo Nacional de Investigaciones Cient\'ificas y T\'ecnicas (CONICET), Argentina, for the financial support. E.K.L. thanks the financial support of the CNPq under Grant No. 303642/2014-9.

\appendix
\section{Limit values of \texorpdfstring{$H_\alpha$}{Halpha} and \texorpdfstring{$C_\alpha$}{Calpha}}
\label{sec:limit-values}
\begin{thm}
Given an arbitrary discrete probability distribution
$p=(p_1,\ldots,p_n)$, let $r$ be the number of
non-zero components of $p$. If $r=1$, then
\begin{enumerate}
\item[\upshape{(i)}] $H_\alpha(p)=0$ and $C_\alpha(p)=0$;
\end{enumerate}
if $r>1$ and $p_M,p_m\in[0,1]$ denote the maximum and minimum
components of $p$ respectively, then
\begin{enumerate}[wide, labelwidth=!, labelindent=0pt]
\item[\upshape{(ii)}]
$\lim_{\alpha\downarrow 0}H_\alpha(p)=\ln r/\ln n$;
\item[\upshape{(iii)}]
$\lim_{\alpha\to\infty}H_\alpha(p)=-\ln p_M/\ln n$;
\item[\upshape{(iv)}]
$\lim_{\alpha\downarrow 0}C_\alpha(p)=\frac{\ln r}{\ln n}\frac{\ln (n+r)-\ln 2n}{\ln (n+1)-\ln 2n}$;
\item[\upshape{(v)}]
$\lim_{\alpha\to\infty}C_\alpha(p)=\frac{\ln[(np_M+1)(np_m+1)]-\ln 4np_M}{\ln 4n-\ln(n+1)}\frac{\ln p_M}{\ln n}$.
\end{enumerate}
\end{thm}

\begin{proof}
Items (i) and (ii) follow directly from the definitions of $H_\alpha$ and $C_\alpha$.
\begin{enumerate}[wide, labelwidth=!, labelindent=0pt]
\item[(iii)] We have $p_M^\alpha\le \sum_{i=1}^np_i^\alpha\le rp_M^\alpha$ and, consequently, $\ln p_M^\alpha\le \ln(\sum_{i=1}^np_i^\alpha)\le\ln(rp_M^\alpha)$. Hence, for $\alpha>1$,
\begin{multline}
\frac{\alpha}{1-\alpha}\frac{\ln p_M}{\ln n}\ge\frac{1}{(1-\alpha)\ln n}\ln\sum_{i=1}^np_i^\alpha\\
\ge \frac{\alpha}{1-\alpha}\frac{\ln p_M}{\ln n}+\frac{\ln r}{(1-\alpha)\ln n}\,.
\end{multline}
Letting $\alpha$ increase without bound, we obtain $H_\alpha(p)\to-\ln p_M/\ln n$.

\item[(iv)] We verify immediately that
\begin{equation}
\lim_{\alpha\downarrow 0}D_\alpha^*=-\frac{1}{2}\ln\frac{n+1}{2n}
\end{equation}
and
\begin{equation}
\lim_{\alpha\downarrow 0} D_\alpha(p)=-\frac{1}{2}\ln\dpar{\frac{1}{2}+\frac{r}{n}}\,.
\end{equation}
Hence,
\begin{equation}
\lim_{\alpha\downarrow 0}C_\alpha(p)=\frac{\ln r}{\ln n}\frac{\ln(n+r)-\ln 2n}{\ln (n+1)-\ln 2n}\,.
\end{equation}
\item[(v)] We have
\begin{equation}
 \label{eq:DD}
\frac{D_\alpha(p)}{D_\alpha^*}=\frac{\ln[\sum_{i=1}^np_i(\frac{p_i+1/n}{2p_i})^{1-\alpha}\sum_{j=1}^n\frac{1}{n}(\frac{p_j+1/n}{2/n})^{1-\alpha}]}{(\alpha-1)\ln\frac{4n}{n+1}+\ln\frac{(n+1)^{1-\alpha}+n-1}{n}}\,.
\end{equation}
To find the limit of Eq. (\ref{eq:DD}) as $\alpha\to\infty$, we start by obtaining convenient upper and lower bounds of $D_\alpha(p)/D_\alpha^*$. With this in mind, we note that
\begin{equation}
\frac{p_i+1/n}{p_i}=1+\frac{1}{np_i}\ge 1+\frac{1}{np_M}
\end{equation}
and, consequently, for $\alpha>1$,
\begin{equation}
\label{eq:DD.up1}
\sum_{i=1}^np_i\dpar{\frac{p_i+1/n}{2p_i}}^{1-\alpha}\le \dpar{\frac{p_M+1/n}{2p_M}}^{1-\alpha}\,.
\end{equation}
On the other hand, using the fact that
\begin{equation}
\frac{p_i+1/n}{1/n}=np_i+1\ge np_m+1\,,
\end{equation}
we obtain that
\begin{equation}
\label{eq:DD.up2}
\sum_{i=1}^n\frac{1}{n}\dpar{\frac{p_i+1/n}{2/n}}^{1-\alpha}\le \dpar{\frac{np_m+1}{2}}^{1-\alpha}\,.
\end{equation}
Since the denominator in Eq.~(\ref{eq:DD}) is positive for sufficiently large $\alpha$, using Eqs.~(\ref{eq:DD.up1}) and~(\ref{eq:DD.up2}) in Eq.~(\ref{eq:DD}), we have
\begin{equation}
\label{eq:DD.ub}
\frac{D_\alpha(p)}{D_\alpha^*}\le \frac{(\alpha-1)\ln\frac{4np_M}{(np_M+1)(np_m+1)}}{(\alpha-1)\ln\frac{4n}{n+1}+\ln\frac{(n+1)^{1-\alpha}+n-1}{n}}
\end{equation}
for sufficiently large $\alpha$. To obtain a lower bound, we note immediately that
\begin{equation}
p_M\dpar{\frac{p_M+1/n}{2p_M}}^{1-\alpha}\le \sum_{i=1}^np_i\dpar{\frac{p_i+1/n}{2p_i}}^{1-\alpha}
\end{equation}
and that
\begin{equation}
\frac{1}{n}\dpar{\frac{p_m+1/n}{2/n}}^{1-\alpha}\le \sum_{i=1}^n\frac{1}{n}\dpar{\frac{p_i+1/n}{2/n}}^{1-\alpha}\,.
\end{equation}
Hence,
\begin{equation}
\label{eq:DD.lb}
\frac{D_\alpha(p)}{D_\alpha^*}\ge \frac{(\alpha-1)\ln\frac{4np_M}{(np_M+1)(np_m+1)}+\ln\frac{p_M}{n}}{(\alpha-1)\ln\frac{4n}{n+1}+\ln\frac{(n+1)^{1-\alpha}+n-1}{n}}
\end{equation}
for sufficiently large values of $\alpha$. Then, it follows from Eqs.~(\ref{eq:DD.ub}) and~(\ref{eq:DD.lb}) that
\begin{equation}
\lim_{\alpha\to\infty}\frac{D_\alpha(p)}{D_\alpha^*}=\frac{\ln 4np_M-\ln[(np_M+1)(np_m+1)]}{\ln 4n-\ln(n+1)}\,.
\end{equation}
The product of this result with the one obtained in item~(iii) is equal to the limit of $C_\alpha(p)$ as $\alpha$ increases without bound.\qedhere
\end{enumerate}
\end{proof}

\bibliographystyle{elsarticle-num}
\bibliography{renyi-entropy-plane.bib}

\end{document}